\documentclass[11pt]{article}
\usepackage{amsfonts}
\usepackage{amsmath}
\usepackage{graphicx}
\usepackage{geometry}

\newtheorem{theorem}{Theorem}

\newtheorem{proposition}[theorem]{Proposition}

\newenvironment{proof}[1][Proof]{\noindent\textbf{#1.} }{\ \rule{0.5em}{0.5em}}

\begin{document}

\title{An Application of the Deutsch-Josza Algorithm to Formal Languages and the Word Problem in Groups}
\author{Michael Batty \and Andrea Casaccino\thanks{%
Corresponding author: Information Engineering Department, University of Siena, Italy; ndr981@tin.it} \and Andrew J. Duncan\thanks{%
Department of Mathematics, University of Newcastle upon Tyne, UK; a.duncan@ncl.ac.uk} \and Sarah Rees\thanks{%
Department of Mathematics, University of Newcastle upon Tyne, UK; Sarah.Rees@ncl.ac.uk} \and Simone Severini\thanks{%
Institute for Quantum Computing and Department of Combinatorics and Optimization, University
of Waterloo, Canada; simoseve@gmail.com}}
\maketitle

\begin{abstract}
We adapt the Deutsch-Josza algorithm to the context of formal
language theory. Specifically, we use the algorithm to
distinguish between trivial and nontrivial words in groups given by
finite presentations, under the promise that a word is of a certain
type. This is done by extending the original algorithm to functions
of arbitrary length binary output and with the introduction of a
more general concept of parity. We provide examples in which properties inherited directly from the original algorithm
allow to reduce the number of oracle queries with respect to
the deterministic classical case. This has some consequences for
the word problem in groups with a particular kind of presentation.
\end{abstract}

\section{The Deutsch-Josza algorithm adapted to formal languages}

We apply a direct generalization of the Deutsch-Josza
algorithm to the context of formal language theory. More
particularly, we adapt the algorithm to distinguish between trivial
and nontrivial words in groups given by finite presentations, under
the promise that a word is of a certain type. For background
information, we refer the reader to \cite{l} and \cite{2}.

The Deutsch-Josza algorithm concerns maps $f:\{0,1\}^{n}\longrightarrow
\{0,1\}$, which we may think of as words of length $n$ in a
two-letter alphabet. Instead, let us consider maps $f:\{0,1\}^{n}\longrightarrow
\{0,1\}^{k}$, where $k$ does not necessarily depend on $n$. Once fixed $k=2$, we can
identify the letters of the alphabet $\mathcal{A}=\{a,b,c,d\}$ with the binary strings of $\{0,1\}^{2}$: $a\leftrightarrow 00$, $b\leftrightarrow 01$, $%
c\leftrightarrow 10$ and $d\leftrightarrow 11$.

 We describe below the simplest possible case: the map $f$ takes a single binary digit as input and gives two binary digits as output. The output corresponds to one of the letters from $\mathcal{A}$. Although this example is not general enough to be interesting, it is still useful to see how the \textquotedblleft balanced VS. constant\textquotedblright question in the original Deutsch-Jozsa task can be lifted to different parities related to the function. This is essentially the same identical quantum circuit implementing the Deutsch algorithm, but with auxiliary input $|11\rangle $ rather than $|1\rangle $. For the sake of clarity, let us see the steps of the algorithm. After applying the Hadamard gates to the two registers, the state of the system is
\begin{equation*}
H\otimes H^{\otimes 2}(|0\rangle \otimes |11\rangle )=|+\rangle \otimes |-\rangle^{\otimes 2}.
\end{equation*}%
If $z\in \{0,1\}$, the oracle works as follows:
\[
\begin{tabular}{ll}
$U_{f}\left( |z\rangle \otimes |-\rangle ^{\otimes 2}\right) $ & $=|z\rangle
\otimes \frac{1}{2}\left( |00\oplus f(z)\rangle -|01\oplus f(z)\rangle
-|10\oplus f(z)\rangle +|11\oplus f(z)\rangle \right) $ \\
& $=(-1)^{p(f(z))}|z\rangle \otimes \left( -|-\rangle \right) ^{\otimes 2}.$%
\end{tabular}%
\]
For a binary string $y$, we denote by $p(y)$ the \emph{parity} of
$y$, that is $p(y)=m(mod2)$, where $m$ is the Hamming weight of $y$.
After querying the oracle $U_{f}$, we obtain the state
\begin{equation*}
\frac{(-1)^{p(f(0))}|0\rangle +(-1)^{p(f(1))}|1\rangle
}{\sqrt{2}}\otimes |-\rangle ^{\otimes 2}.
\end{equation*}%
Finally, after the last Hadamard gate, the first qubit will be in the
state $|0\rangle$ if $p(f(0))=p(f(1))$ or $|1\rangle$ if
$p(f(0))\neq p(f(1))$. We shall say that $f$ is \emph{parity
constant} if $p(f(0))=p(f(1))$; \emph{parity balanced},
otherwise. By measuring the final state, we
obtain $|0\rangle $ with probability $1$ if $f$ is parity balanced and $%
|1\rangle $ with probability $1$ if $f$ is parity constant. In the same spirit, moving to a larger number of bits, a function $f$ is parity balanced if exactly half of the elements of the image of $f$ have odd parity. We will show that properties of the Deutsch-Jozsa algorithm are inherited when extending the co-domain of $f$ and generalizing the notion of parity in less trivial ways.

\bigskip

Let us now introduce some terminology related to formal languages. Given a word $%
w:\{0,1\}^{n}\longrightarrow \{a,b,c,d\}$, an \emph{anagram} of $w$ is a
word of the form $w\circ \phi $, where $\phi :\{0,1\}^{n}\longrightarrow
\{0,1\}^{n}$ is a permutation. We write $[w]$ for the set of all anagrams of
$w$. More formally, let $F$ denote the free monoid on $\{a,b,c,d\}$ and let $%
M$ denote the free commutative monoid on $\{a,b,c,d\}$. Let $R$
denote the natural map from $F$ to $M$ and suppose that $w\in M$.
Then $R(w)=[w]$, the set of all anagrams of $w$. It is clear that
the definition of parity balanced and parity constant
extends to the words of $M$. Let $x\in \{01,10,11\}$, we denote the sets of $x$\emph{-constant} and $x$\emph{-balanced} words of
length $k$ over $\mathcal{A}$ by $\mathcal{C}_{k}^{x}(\mathcal{A})$ and $\mathcal{B}_{k}^{x}(\mathcal{A})$, respectively. The set of $11$-constant words is
then a union of sets of anagrams $$\mathcal{C}_{2}^{11}(a,b,c,d)=[aa]\cup \lbrack bb]\cup \lbrack
cc]\cup \lbrack dd]\cup \lbrack bc]\cup \lbrack ad].$$
Similarly, the set of $11$-balanced words is
$$\mathcal{B}_{2}^{11}(a,b,c,d)=[ab]\cup \lbrack ac]\cup \lbrack
bd]\cup \lbrack cd].$$ Note that both the terms of the alphabet in the bracket have the same parity
with the notation $a\leftrightarrow 00$, $b\leftrightarrow 01$, $%
c\leftrightarrow 10$ and $d\leftrightarrow 11$.

Suppose not to input $|11\rangle $ into the auxiliary
workspace, but rather some arbitrary string of length two. How does this
affect the sets of words we can distinguish between? It is
interesting to observe that we may define as follows a more general type of
parity. The set $\{00,01,10,11\}$ is
considered in natural way as the vector space $(\mathbb{Z}_{2})^{2}=\mathbb{Z}%
_{2}\oplus \mathbb{Z}_{2}$. Define $p^{x}(y)$ to be equal to $0$,
if $y$ is in the subspace $\langle x\rangle =\{00,x\}$ and equal
to $1$, otherwise. With this notation, $p^{11}(y)=p(y)$, the usual
parity function. A similar circuit, taking the auxiliary input $\neg(x)$, that is the binary complement of $x$,
will distinguish between whether the word is $x$%
-constant or $x$-balanced. Again, measurement of the state will yield this
information with certainty. It is clear that if $x=00$ then the
output of the circuit is
independent of $f$, and so this is of no use. Let us now suppose that $x=01$. Then $x$%
-constant means that the outputs of $f$ are in the same coset of the subgroup
$\{00,x=01\}$ in $(\mathbb{Z}_{2})^{2}$ and $x$-balanced means that
$f(0)$ and $f(1)$ are in different cosets, or, in other words, both
in or out the subspace $\langle x\rangle =\{00,x\}$. The set of
$01$-constant words is
$$\mathcal{C}_{2}^{01}(a,b,c,d)=[aa]\cup
\lbrack bb]\cup \lbrack cc]\cup \lbrack dd]\cup \lbrack ab]\cup
\lbrack cd]$$
and the set of $01$-balanced words is
$$\mathcal{B}_{2}^{01}(a,b,c,d)=[ac]\cup \lbrack ad]\cup \lbrack
bc]\cup \lbrack bd].$$
With the same notation,
$$\mathcal{C}_{2}^{10}(a,b,c,d)=[aa]\cup \lbrack bb]\cup \lbrack
cc]\cup \lbrack dd]\cup \lbrack ac]\cup \lbrack bd]$$
and
$$\mathcal{B}_{2}^{10}(a,b,c,d)=[ab]\cup \lbrack ad]\cup \lbrack
bc]\cup \lbrack cd].$$
As before, the first term and the second term in the bracket
represent the first output and the second output of the function, respectively.
Also the parity is the same as described before. Note that when the set is parity
constant both terms are in or out the subspace $\langle x\rangle$,
while in the parity balanced case one term is in the subspace and
the other one is out.

We write $$\mathcal{F}%
_{k}^{x}(\mathcal{A})=\mathcal{C}_{k}^{x}(\mathcal{A})\cup \mathcal{B}%
_{k}^{x}(\mathcal{A})$$ and call this the set of $x$-\emph{feasible} words of
length $k$. Notice that $$\mathcal{A}^{k}=\bigcup\nolimits_{x} \mathcal{F}%
_{k}^{x}(\mathcal{A}).$$

The following fact is central in the context of our discussion.

\begin{theorem}
Fixed an $x$-parity, we can decide if a function $f$ is $x$-constant or $x$-balanced with a single quantum query. Equivalently, we can determine if the output of $f$ is a language in $\mathcal{C}_{k}^{x}$ or $\mathcal{B}_{k}^{x}$, with a single quantum query.
\end{theorem}

Already in the seminal work \cite{3}, it was pointed out that a classical randomized algorithm solves the Deutsch-Josza task with three classical queries on average, whereas the quantum approach solves it with probability 1 using one single query (see also \cite{4}). Here the output of the function $f$ is no more a single bit but a bit string. If the number of letters of the alphabet is $d$ then the output of the function is an ${n}$-bit string, where $n=\log_{2}d$. A word is given by $k$ repeated random output of the function, where $k$ is the length of the word. In other terms, a word is like a sequence obtained by tossing a dice with $d$ faces. It is easy to see that the probability of being constant over all possible anagrams, interpreting the output binary string of $k$ queries as anagrams of $k$ letters, is higher than in the balanced case. The difference decreases while increasing the number of queries.
As long as any possible parity function partitions into two classes the function co-domain, the number of quantum queries required to distinguish between the parity constant and parity balanced cases remains constant. This is due to the fact that binary strings always form a bipartition with respect to the Hamming weight.
The method described allows us to extend the Deutsch-Josza algorithm to functions with output of any dimension,  $f:\{0,1\}^{n}\longrightarrow \{0,1\}^{k}$. Defining appropriate parities, based on subgroups or code membership problems, could give arise to potentially interesting applications.

\section{Distinguishing between languages}

In this section we construct languages given by intersecting the images of binary maps. We show that acceptance of a word of length $k$ in one of these languages can be determined with $k$ quantum queries. This can be easily done on the basis of the discussion carried on in the previous section. The problem defined is artificial, but nonetheless indicates a way to use repeated applications of the modified Deutsch-Jozsa algorithm, with special reference to formal languages. We define languages constructed by intersecting the images of functions promised to be $x$-constant or $x$-balanced with respect to different subspaces. If $X \subset \{01,10,11\}$, let us write $$\mathcal{F}_{k}^{X}(%
\mathcal{A}) := \bigcap _{x\in
X}\mathcal{F}_{k}^{x}(\mathcal{A}).$$ We have

\begin{eqnarray*}
\mathcal{C}_{2}^{11}(a,b,c,d) &=&[aaaa]\cup \lbrack bbbb]\cup \lbrack
cccc]\cup \lbrack dddd]\cup \lbrack aaad] \\
&&\cup \lbrack aadd]\cup \lbrack addd]\cup \lbrack
bbbc]\cup \lbrack bbcc]\cup \lbrack bccc], \\
\mathcal{B}_{2}^{11}(a,b,c,d) &=&[aabb]\cup \lbrack aacc]\cup \lbrack
bbdd]\cup \lbrack ccdd] \\
&&\cup \lbrack aabc]\cup \lbrack bcdd]\cup \lbrack abbd]\cup \lbrack
accd]\cup \lbrack abcd], \\
\mathcal{C}_{2}^{01}(a,b,c,d) &=&[aaaa]\cup \lbrack bbbb]\cup \lbrack
cccc]\cup \lbrack dddd]\cup \lbrack aaab] \\
&&\cup \lbrack aabb]\cup \lbrack abbb]\cup \lbrack
cccd]\cup \lbrack ccdd]\cup \lbrack cddd], \\
\mathcal{B}_{2}^{01}(a,b,c,d) &=&[aacc]\cup \lbrack aadd]\cup \lbrack
bbcc]\cup \lbrack bbdd] \\
&&\cup \lbrack aacd]\cup \lbrack bbcd]\cup \lbrack abcc]\cup \lbrack
abdd]\cup \lbrack abcd], \\
\mathcal{F}_{2}^{\{01,11\}} &=&[aaaa]\cup \lbrack bbbb]\cup \lbrack
cccc]\cup \lbrack dddd]\cup \lbrack aabb] \\
&&\cup \lbrack aacc]\cup \lbrack aadd]\cup \lbrack bbcc]\cup \lbrack
bbdd]\cup \lbrack ccdd]\cup \lbrack abcd].
\end{eqnarray*}

We also have
\begin{eqnarray*}
\mathcal{B}_{2}^{11}(a,b,c,d)\cap \mathcal{B}_{2}^{01}(a,b,c,d)
&=&[abcd]\cup \lbrack aacc]\cup \lbrack bbdd], \\
\mathcal{C}_{2}^{11}(a,b,c,d)\cap \mathcal{B}_{2}^{01}(a,b,c,d)
&=&[aadd]\cup \lbrack bbcc], \\
\mathcal{B}_{2}^{11}(a,b,c,d)\cap \mathcal{C}_{2}^{01}(a,b,c,d)
&=&[aabb]\cup \lbrack ccdd], \\
\mathcal{C}_{2}^{11}(a,b,c,d)\cap \mathcal{C}_{2}^{01}(a,b,c,d)
&=&[aaaa]\cup \lbrack bbbb]\cup \lbrack cccc]\cup \lbrack dddd].
\end{eqnarray*} Therefore, given a word in
$\mathcal{F}_{2}^{\{01,11\}}$, we can decide with two quantum
queries in which of these four languages the word is. This is an improvement over the classical
deterministic setting, where we need at least $2^{n}-1$ queries for each function. The remaining
possibilities for $x$ are
\begin{eqnarray*}
\mathcal{C}_{2}^{10}(a,b,c,d) &=&[aaaa]\cup \lbrack bbbb]\cup \lbrack
cccc]\cup \lbrack dddd]\cup \lbrack aaac] \\
&&\cup \lbrack aacc]\cup \lbrack accc]\cup \lbrack bbbd]\cup \lbrack
bbdd]\cup \lbrack bddd], \\
\mathcal{B}_{2}^{10}(a,b,c,d) &=&[aabb]\cup \lbrack aadd]\cup \lbrack
bbcc]\cup \lbrack ccdd]\cup \lbrack aabd] \\
&&\cup \lbrack bccd]\cup \lbrack abbc]\cup \lbrack acdd]\cup \lbrack abcd],
\\
\mathcal{F}_{2}^{\{10,11\}} &=&[aaaa]\cup \lbrack bbbb]\cup \lbrack
cccc]\cup \lbrack dddd]\cup \lbrack aabb]\cup \lbrack aacc] \\
&&\cup \lbrack aadd]\cup \lbrack bbcc]\cup \lbrack
bbdd]\cup \lbrack ccdd]\cup \lbrack abcd].
\end{eqnarray*} We then have
$$\mathcal{F}_{2}^{\{01,11\}}=\mathcal{F}_{2}^{\{10,11\}}.$$ It can
be checked that this is also equal to
$\mathcal{F}_{2}^{\{01,10\}}$.
However, the three possibilities $X=\{01,11\},\{10,11\}$ and
$\{01,10\}$ all distinguish between different languages, since we
have
\begin{eqnarray*}
\mathcal{B}_{2}^{11}(a,b,c,d)\cap \mathcal{B}_{2}^{10}(a,b,c,d)
&=&[abcd]\cup \lbrack aabb]\cup \lbrack ccdd], \\
\mathcal{C}_{2}^{11}(a,b,c,d)\cap \mathcal{B}_{2}^{10}(a,b,c,d)
&=&[aadd]\cup \lbrack bbcc], \\
\mathcal{B}_{2}^{11}(a,b,c,d)\cap \mathcal{C}_{2}^{10}(a,b,c,d)
&=&[aacc]\cup \lbrack bbdd], \\
\mathcal{C}_{2}^{11}(a,b,c,d)\cap \mathcal{C}_{2}^{10}(a,b,c,d)
&=&[aaaa]\cup \lbrack bbbb]\cup \lbrack cccc]\cup \lbrack dddd], \\
\mathcal{B}_{2}^{01}(a,b,c,d)\cap \mathcal{B}_{2}^{10}(a,b,c,d)
&=&[abcd]\cup \lbrack aadd]\cup \lbrack bbcc], \\
\mathcal{C}_{2}^{01}(a,b,c,d)\cap \mathcal{B}_{2}^{10}(a,b,c,d)
&=&[aabb]\cup \lbrack ccdd], \\
\mathcal{B}_{2}^{01}(a,b,c,d)\cap \mathcal{C}_{2}^{10}(a,b,c,d)
&=&[aacc]\cup \lbrack bbdd], \\
\mathcal{C}_{2}^{01}(a,b,c,d)\cap \mathcal{C}_{2}^{10}(a,b,c,d)
&=&[aaaa]\cup \lbrack bbbb]\cup \lbrack cccc]\cup \lbrack dddd].
\end{eqnarray*}

We have then seen that $$\mathcal{F}_{2}^{\{01,11\}}=\mathcal{F}_{2}^{\{10,11\}}=%
\mathcal{F}_{2}^{\{01,10\}}=\mathcal{F}_{2}^{\{01,10,11\}}.$$ This fact will be useful later, when dealing with the word problem.

\section{Larger alphabets}

A similar approach can be taken for larger alphabets:
$$\{a,b,c,d,e,f,g,h\}\rightarrow\{000,001,010,011,100,101,110,111\}.$$
The previous treatment applies in a straightforward manner. It is in fact still possible to define a parity, based on the even
number of 1s, like $p^{11}$. This is equivalent to determine if a
word $w$ is in the subspace \{000,011,101,110\}, also denoted
$p^{adfg}$. In this case, the set of parity constant and parity
balanced words can be obtained using the auxiliary input
$|111\rangle $ in the circuit described before:
\begin{eqnarray*}
&&U_{f}\left( |z\rangle \otimes |-\rangle ^{\otimes 3}\right) \\
&=&|z\rangle \otimes \frac{1}{2}(|000\oplus f(z)\rangle -|001\oplus
f(z)\rangle -|010\oplus f(z)\rangle +|011\oplus f(z)\rangle \\&-&|100\oplus f(z)\rangle +|101\oplus f(z)\rangle +|110\oplus
f(z)\rangle -|111\oplus f(z)\rangle) \\
&=&(-1)^{p(f(z))}|z\rangle \otimes \left(-|-\rangle \right) ^{\otimes 3}.
\end{eqnarray*}%
Then $U_{f}$ gives the following set:
\begin{eqnarray*}
\mathcal{C}_{2}^{adfg}(a,b,c,d,e,f,g,h) &=&[aa]\cup \lbrack bb]\cup \lbrack
cc]\cup \lbrack dd]\cup \lbrack ee]\cup \lbrack ff]\cup \lbrack gg]\cup \lbrack hh]\cup \lbrack ad]\cup \lbrack af]\\
&&\cup \lbrack ag]\cup \lbrack df]\cup \lbrack dg]\cup \lbrack fg]\cup \lbrack bc]\cup \lbrack be]\cup \lbrack bh]\cup
\lbrack ce]\cup \lbrack ch]\cup \lbrack eh].
\end{eqnarray*}
Similarly, the set of parity balanced words is
\begin{eqnarray*}
\mathcal{B}_{2}^{adfg}(a,b,c,d,e,f,g,h) &=&[ab]\cup \lbrack ac]\cup \lbrack
bd]\cup \lbrack cd]\cup \lbrack ae]\cup \lbrack ah]\cup \lbrack de]\cup \lbrack dh]\\
&&\cup \lbrack bf]\cup \lbrack bg]\cup \lbrack cf]\cup \lbrack cg]\cup \lbrack fe]\cup \lbrack fh]\cup \lbrack ge]\cup
\lbrack gh].
\end{eqnarray*}

Other parities can be defined considering different set of
vectors. For our purposes it is sufficient to define a set composed
by the elements
$p^{abcd}=\{000,001,010,011\}$. This plays the same role as $p^{01}$. In this
case, the set of parity constant word can be obtained by using $|100\rangle $ as
auxiliary input. The circuit has the following output:
\begin{eqnarray*}
&&U_{f}\left(|z\rangle \otimes |-\rangle \otimes |+\rangle ^{\otimes 2}\right) \\
&=&|z\rangle \otimes \frac{1}{2}( |000\oplus f(z)\rangle +|001\oplus
f(z)\rangle +|010\rangle \oplus f(z)\rangle +|011\oplus f(z)\rangle \\ &-&|100\oplus
f(z)\rangle-|101\oplus
f(z)\rangle -|110\oplus
f(z)\rangle -|111\oplus
f(z)\rangle) \\
&=&(-1)^{p(f(z))}|z\rangle \otimes \left(-|-\rangle \right) \otimes |+\rangle ^{\otimes 2}.
\end{eqnarray*}%
As we have said before, this procedure gives
\begin{eqnarray*}
\mathcal{C}_{2}^{abcd}(a,b,c,d,e,f,g,h) &=&[aa]\cup \lbrack bb]\cup \lbrack
cc]\cup \lbrack dd]\cup \lbrack ee]\cup \lbrack ff]\cup \lbrack gg]\cup \lbrack hh]\cup \lbrack ab]\cup \lbrack ac] \\
&&\cup\lbrack ad]\cup \lbrack bc]\cup \lbrack bd]\cup \lbrack cd]\cup \lbrack ef]\cup \lbrack eg]\cup \lbrack eh]\cup
\lbrack fg]\cup \lbrack fh]\cup \lbrack gh]
\end{eqnarray*}%
and
\begin{eqnarray*}
\mathcal{B}_{2}^{abcd}(a,b,c,d,e,f,g) &=&[ae]\cup \lbrack af]\cup \lbrack
ag]\cup \lbrack ah]\cup \lbrack be]\cup \lbrack bf]\cup \lbrack bg]\cup \lbrack bh]\\
&&\cup \lbrack ce]\cup \lbrack cf]\cup \lbrack cg]\cup \lbrack ch]\cup \lbrack de]\cup \lbrack df]\cup \lbrack dg]\cup \lbrack dh].
\end{eqnarray*}
For reasons that will be clear later, it is important to define
also the parity, based on the subspace
$p^{adeh}=\{000,011,101,111\}$, for which the set of parity
constant words is obtained by setting as auxiliary input the state
$|011\rangle $:
\begin{eqnarray*}
&&U_{f}\left(|z\rangle \otimes |+\rangle \otimes |-\rangle ^{\otimes 2} \right) \\
&=&|z\rangle \otimes \frac{1}{2}( |000\oplus f(z)\rangle -|001\oplus
f(z)\rangle -|010\rangle \oplus f(z)\rangle +|011\oplus f(z)\rangle+ \\ &-&|100\oplus
f(z)\rangle+|101\oplus
f(z)\rangle -|110\oplus
f(z)\rangle +|111\oplus
f(z)\rangle) \\
&=&(-1)^{p(f(z))}|z\rangle \otimes \left(|-\rangle \right)^{\otimes 2} \otimes |+\rangle.
\end{eqnarray*}
The sets produced are
\begin{eqnarray*}
\mathcal{C}_{2}^{adeh}(a,b,c,d,e,f,g,h) &=&[aa]\cup \lbrack bb]\cup \lbrack
cc]\cup \lbrack dd]\cup \lbrack ee]\cup \lbrack ff]\cup \lbrack gg]\cup \lbrack hh]\cup \lbrack ad]\cup \lbrack ae]\\
&&\cup \lbrack ah]\cup \lbrack de]\cup \lbrack dh]\cup \lbrack eh]\cup \lbrack bc]\cup \lbrack bf]\cup \lbrack bg]\cup
\lbrack cf]\cup \lbrack cg]\cup \lbrack fg];
\end{eqnarray*}
for the balanced case, we have
\begin{eqnarray*}
\mathcal{B}_{2}^{adeh}(a,b,c,d,e,f,g,h) &=&[ab]\cup \lbrack ac]\cup \lbrack
af]\cup \lbrack ag]\cup \lbrack db]\cup \lbrack dc]\cup \lbrack df]\cup \lbrack dg]\\
&&\cup \lbrack eb]\cup \lbrack ec]\cup \lbrack ef]\cup \lbrack eg]\cup \lbrack hb]\cup \lbrack hc]\cup \lbrack hf]\cup
\lbrack hg].
\end{eqnarray*}

It is indeed possible to generalize the circuit for an arbitrary
length binary function co-domain. In particular, the length of the
output binary string will be determined by the logarithm of the
cardinality of the alphabet considered (for example, two bits for a 4-elements
alphabet). Moreover, to each parity function subspace corresponds
a unique input to be fed into the circuit shown before. The Hadamard gate transforms each qubit of the input binary string
into the state $|+\rangle$ or $|-\rangle$ depending on the value of
the qubit. For the generic input
$|0\ldots1\rangle$, we have

\begin{eqnarray*}
&&U_{f}\left(|z\rangle \otimes |+\rangle \otimes \ldots \otimes |+\rangle \otimes |-\rangle \right)=(-1)^{p(f(x))}|z\rangle \otimes |+\rangle ^{\otimes n+1}.
\end{eqnarray*}

If $k=4$,
for $p^{adfg}$, the set of parity balanced
words is
\begin{eqnarray*}
\mathcal{C}_{4}^{adfg}(a,b,c,d,e,f,g,h) &=&[aaaa]\cup \lbrack bbbb]\cup
\lbrack cccc]\cup \lbrack dddd]\cup \lbrack eeee] \\
&&\cup \lbrack ffff]\cup \lbrack gggg]\cup \lbrack hhhh]\cup \lbrack
aaad]\cup \lbrack aadd] \\
&&\cup \lbrack addd]\cup \lbrack aaaf]\cup \lbrack aaff]\cup \lbrack
afff]\cup \lbrack aaag] \\
&&\cup \lbrack aggg]\cup \lbrack dddf]\cup \lbrack ddff]\cup \lbrack
dfff]\cup \lbrack fffg] \\
&&\cup \lbrack ffgg]\cup \lbrack fggg]\cup \lbrack bbbc]\cup \lbrack
bbcc]\cup \lbrack bccc] \\
&&\cup \lbrack bbbe]\cup \lbrack bbee]\cup \lbrack beee]\cup \lbrack
ccce]\cup \lbrack ccee] \\
&&\cup \lbrack ceee]\cup \lbrack bbbh]\cup \lbrack bbhh]\cup \lbrack
bhhh]\cup \lbrack ccch] \\
&&\cup \lbrack cchh]\cup \lbrack chhh]\cup \lbrack eeeh]\cup \lbrack
eehh]\cup \lbrack ehhh];
\end{eqnarray*}
while the set of parity balanced words is
\begin{eqnarray*}
\mathcal{B}_{4}^{adfg}(a,b,c,d,e,f,g,h) &=&[aabb]\cup \lbrack aacc]\cup
\lbrack aaee]\cup \lbrack aahh]\cup \lbrack ddbb] \\
&&\cup \lbrack ddcc]\cup \lbrack ddee]\cup \lbrack ddhh]\cup \lbrack
ffbb]\cup \lbrack ffcc] \\
&&\cup \lbrack ffee]\cup \lbrack ffhh]\cup \lbrack ggbb]\cup \lbrack
ggcc]\cup \lbrack ggee] \\
&&\cup \lbrack gghh]\cup \lbrack adbc]\cup \lbrack afce]\cup \lbrack
agbc]\cup \lbrack agbe] \\
&&\cup \lbrack agce]\cup \lbrack adce]\cup \lbrack adbe]\cup \lbrack
adhe]\cup \lbrack agch] \\
&&\cup \lbrack afce]\cup \lbrack afch]\cup \lbrack adbe]\cup \lbrack
adbh]\cup \lbrack afbc] \\
&&\cup \lbrack afbh]\cup \lbrack agbh]\cup \lbrack ageh]\cup \lbrack
afeh]\cup \lbrack afbe].
\end{eqnarray*}

For $p^{abcd}$, we have
\begin{eqnarray*}
\mathcal{C}_{4}^{abcd}(a,b,c,d,e,f,g,h) &=&[aaaa]\cup \lbrack bbbb]\cup
\lbrack cccc]\cup \lbrack dddd]\cup \lbrack eeee] \\
&&\cup \lbrack ffff]\cup \lbrack gggg]\cup \lbrack hhhh]\cup \lbrack
aaab]\cup \lbrack aabb] \\
&&\cup \lbrack abbb]\cup \lbrack aaac]\cup \lbrack aacc]\cup \lbrack
accc]\cup \lbrack aaad] \\
&&\cup \lbrack aadd]\cup \lbrack addd]\cup \lbrack bbbc]\cup \lbrack
bbcc]\cup \lbrack bccc] \\
&&\cup \lbrack bbbd]\cup \lbrack bbdd]\cup \lbrack bddd]\cup \lbrack
cccd]\cup \lbrack ccdd] \\
&&\cup \lbrack cddd]\cup \lbrack eeef]\cup \lbrack eeff]\cup \lbrack
efff]\cup \lbrack eeeg] \\
&&\cup \lbrack eegg]\cup \lbrack eggg]\cup \lbrack eeeh]\cup \lbrack
eehh]\cup \lbrack ehhh] \\
&&\cup \lbrack fffg]\cup \lbrack ffgg]\cup \lbrack fffh]\cup \lbrack
ffhh]\cup \lbrack fhhh] \\
&&\cup \lbrack gggh]\cup \lbrack gghh]\cup \lbrack ghhh]
\end{eqnarray*}%
and
\begin{eqnarray*}
\mathcal{B}_{4}^{abcd}(a,b,c,d,e,f,g,h) &=&[aaee]\cup \lbrack aaff]\cup
\lbrack aagg]\cup \lbrack aahh]\cup \lbrack bbee] \\
&&\cup \lbrack bbff]\cup \lbrack bbgg]\cup \lbrack bbhh]\cup \lbrack
ccee]\cup \lbrack ffcc] \\
&&\cup \lbrack ccgg]\cup \lbrack cchh]\cup \lbrack ddee]\cup \lbrack
ddff]\cup \lbrack ddgg] \\
&&\cup \lbrack ddhh]\cup \lbrack abef]\cup \lbrack abeg]\cup \lbrack
abeh]\cup \lbrack acef] \\
&&\cup \lbrack aceg]\cup \lbrack aceh]\cup \lbrack adef]\cup \lbrack
adeg]\cup \lbrack adeh] \\
&&\cup \lbrack abfg]\cup \lbrack abfh]\cup \lbrack acfg]\cup \lbrack
acfh]\cup \lbrack adfg] \\
&&\cup \lbrack adfh]\cup \lbrack agbh]\cup \lbrack agch]\cup \lbrack adgh].
\end{eqnarray*}

The same reasoning carried on for a four-letter alphabet can be applied to form the set of words $$\mathcal{F}%
_{k}^{x}(\mathcal{A})=\mathcal{C}_{k}^{x}(\mathcal{A})\cup \mathcal{B}%
_{k}^{x}(\mathcal{A})$$ and the relative intersections. A potential
generalization could arise in the context of error correcting codes. This
could be based on introducing an encoding in which the letters of the
alphabet are associated to the codewords of a subspace quantum error correcting code. A
form of parity could be defined by considering the remaining subspaces.

\section{Applications to the word problem in groups}

Let $\{a,b,c=B,d=A\}$ be a paired alphabet, where $A$ represents
$a^{-1}$ and $B$ represents $b^{-1}$. We first consider words of
length $2$. Parity constant words are \textquotedblleft character
constant\textquotedblright , \emph{i.e.} consist of only one letter,
whether it be lower or upper case. Parity balanced words are
\textquotedblleft character balanced\textquotedblright . The words
corresponding to the parity constant case are $aa$,$aA$,$bb$,$bB$,$Bb$,$BB$,$%
Aa$,$AA$. Those corresponding to the parity balanced case are $ab$,$aB$,$%
ba$,$bA$,$Ba$,$BA$,$Ab$,$AB$. The words $w$ in the first list all
satisfy ${w}\in \langle a\rangle \cup \langle b\rangle $ (in fact
we have ${w}\in \langle a^{2}\rangle \cup \langle b^{2}\rangle $),
whereas those $w$ in the second list all satisfy ${w}\notin
\langle a\rangle \cup \langle b\rangle $. Thus, for words of
length $2$, we can determine with a single measurement whether or
not ${w}\in \langle
a\rangle \cup \langle b\rangle $.

If $x=01$ then the $x$-constant words are $%
aa$,$ab$,$ba$,$bb$,$BA$,$BB$,$AA$,$AB$ and the $x$-balanced words are $%
aB $,$aA$,$bB$,$bA$,$Aa$,$Ab$,$Ba$,$Bb$. So, $01$-constant and $01$%
-balanced may be thought of as \textquotedblleft case
constant\textquotedblright\ and \textquotedblleft case
balanced\textquotedblright\, where the case can be upper or lower. For
example, a commutator word (reduced or not) is always case balanced.
\textquotedblleft Case constant\textquotedblright\ and \textquotedblleft
case balanced\textquotedblright\ are properties of $\overline{w}$, rather
than $w$. This is not the case for \textquotedblleft parity
constant\textquotedblright\ and \textquotedblleft parity
balanced\textquotedblright.

If $x=10$ then the $x$-constant words are $aa$,$%
aB$,$bb$,$bA$,$Ba$,$BB$,$Ab$,$AA$ and the $x$-balanced words are $ab$,$%
aA $,$ba$,$bB$,$Bb$,$BA$,$AB$,$Aa$. This does not seem to have any
nice interpretation. The 10-balanced corresponds to the
cyclic subgroup generated by $ab$ and the 11-constant set solves a
problem of union of subgroup membership for $\langle a\rangle \cup
\langle b\rangle $. The elements represented by these words are
depicted on the following Cayley graph portions:

\begin{center}\includegraphics{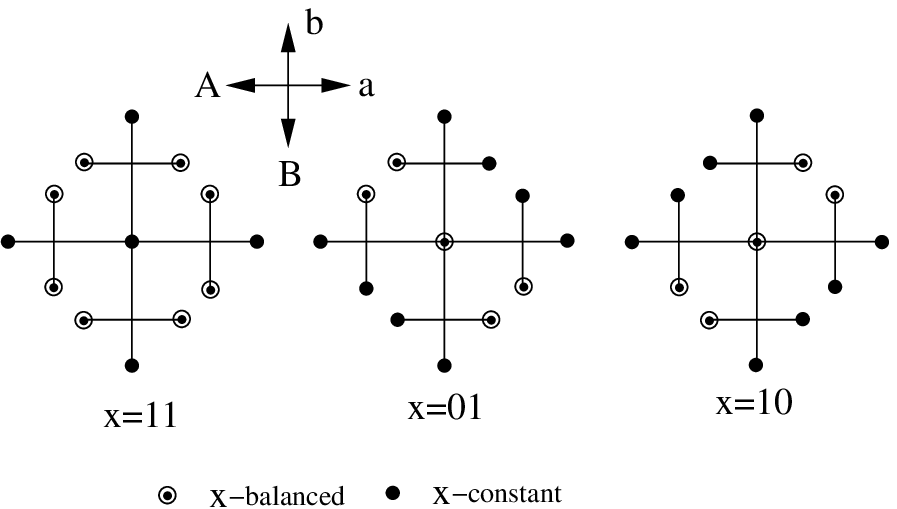}
\end{center}

Note that $w$ is $11$-constant but not $01$-constant; also $w$ is $%
11$-constant and not $10$-constant. Then $w=\mathcal{F}_{2}^{11}$.
This gives a method of solving the word problem for words of length $%
2$ using two quantum queries.

For $k=2$, if we are promised that $w$ is $x$-feasible then the
quantum query complexity of the property ``is $w$ trivial?'' seems to be $2$. 
But this is not a reduction in complexity from the classical case.
However, there is hope that an analogous method might be an improvement in
quantum query complexity for longer words. We have the following:

\begin{proposition}
For all $n$, if we are promised that the word $w$ of length $2^{n}$ is $11$%
-feasible then the quantum query complexity of the property
\textquotedblleft Does $w$ represent an element of $\langle a\rangle \cup
\langle b\rangle $?\textquotedblright\ is 1.
\end{proposition}

This is directly analogous to the Deutsch-Josza algorithm, and the proof is
the same. It is unclear how to extend this approach for the word problem
beyond two letters. Here are examples of two groups where we require
different promises:

\begin{proposition}
Consider the free abelian group $G=\langle a,b\mid ab=ba\rangle $. Let $w$
be a four-letter word in $\mathcal{A}$ which is in $\mathcal{F}_{2}^{11}\cap
\mathcal{F}_{2}^{01}\cap \mathcal{F}_{2}^{10}$. Then the quantum query
complexity of the question \textquotedblleft Does $w$ represent the trivial
element of $G$?\textquotedblright\ is at most $3$.
\end{proposition}

\begin{proof}
The first query asks whether $w\in \mathcal{C}_{2}^{01}$ or $w\in \mathcal{B}%
_{2}^{01}$. If the former is true then $w$ is not trivial so stop. If $w\in
\mathcal{B}_{2}^{01}$ then proceed to the second query, which is whether $%
w\in \mathcal{C}_{2}^{11}$ or $w\in \mathcal{B}_{2}^{11}$. If the former is
true then $w$ is trivial so stop. Otherwise we know that $w\in \mathcal{B}%
_{2}^{11}\cap B_{2}^{01}$ and we may proceed to the third query. There are
two possibilities. The first possibility is that we have a word with two $A$%
s and two $bs$ or a word with two $a$s and two $B$s. That is, $w$ is a
cyclic rotation of $(AAbb)^{\pm 1}$. The second possibility is that we have
one each of $A$, $b$, $a$ and $B$. In the first case, $w$ is nontrivial and
in $\mathcal{C}_{2}^{10}$; in the second case, $w$ is trivial and in $\mathcal{%
B}_{2}^{01}$. So our third query is whether $w\in \mathcal{C}_{2}^{10}$ or $%
w\in \mathcal{B}_{2}^{10}$; this solves the word problem provided $w$ is as
promised.
\end{proof}

\bigskip

It is indeed possible to generalize this theorem to the 8-letters
alphabet introduced earlier, by considering the four-paired alphabet
$\{a,b,c,d,e=D,f=C,g=B,h=A\}$, where the upper-case $A,B,C,D$
letters represent respectively $a^{-1},b^{-1},c^{-1},d^{-1}$. In
particular, we have the following statement:

\begin{proposition}
Consider the free group $G=\langle a,b,c,d\mid abcd=dcba\rangle $. Let $w$
be a 8-letter word in $\mathcal{A}$ which is in $\mathcal{F}_{3}^{adfg}\cap
\mathcal{F}_{3}^{abcd}\cap \mathcal{F}_{3}^{adeh}$. Then the quantum query
complexity of the question \textquotedblleft Does $w$ represent the trivial
element of $G$?\textquotedblright\ is at most $3$.
\end{proposition}

\begin{proof}
The first query asks whether $w\in \mathcal{C}_{3}^{abcd}$ or $w\in \mathcal{B}%
_{3}^{abcd}$. If the former is true then $w$ is not trivial so stop. If $w\in
\mathcal{B}_{3}^{abcd}$ then proceed to the second query, which is whether $%
w\in \mathcal{C}_{3}^{adfg}$ or $w\in \mathcal{B}_{3}^{adfg}$. If the former is
true then $w$ is trivial so stop. Otherwise we know that $w\in \mathcal{B}%
_{3}^{adfg}\cap B_{3}^{abcd}$ and we may proceed to the third query. There are
two possibilities. The first possibility is that we have a word with two $A$%
s two $D$s and two $as$ and two $ds$  or a word with two $C$s two $B$s, two $c$s and two $b$s. That is, $w$ is a
cyclic rotation of $(AADDaadd)^{\pm 1}$ or  $(BBCCbbcc)^{\pm 1}$. The second possibility is that we have
one each of $A$, $b$, $a$ $B$, $C$, $d$,$c$, and $D$ . In the first case, $w$ is nontrivial and
in $\mathcal{C}_{3}^{adeh}$; in the second, $w$ is trivial and in $\mathcal{%
B}_{3}^{abcd}$. Our third query is whether $w\in \mathcal{C}_{3}^{adeh}$ or $%
w\in \mathcal{B}_{3}^{abcd}$; this solves the word problem provided $w$ is as
promised.
\end{proof}

\bigskip

Looking at the first two queries it seems possible to generalize
this result for every paired alphabet of dimension $2^{n-1}$ and
words of length $2^{n}$, by defining parities based on the even
number of ones, like $p^{adfg}$. This is always possible because of
the equipartition of the binary strings with respect to Hamming weight. The last parity required is the one
used to identify words that are cyclic permutations of elements of
the alphabet, for example, $p^{adeh}$.
It does not seem easy to distinguish between trivial and nontrivial
four-letter words in the free group of rank $2$ using less than $4$
quantum queries. However, the first indication that classical query
complexity can be improved upon in a nonabelian finitely presented
group is the following:

\begin{proposition}
Consider the group presented by $G=\langle a,b\mid a^{2}=b^{2}\rangle $.
Suppose we are given a word $w$ of length $4$ in $\mathcal{A}$ such that $%
w\in \mathcal{F}_{2}^{11}\cap \mathcal{F}_{2}^{01}$. Then the quantum query
complexity of the question \textquotedblleft Does $w$ represent the trivial
element of $G$?\textquotedblright\ is at most $3$.
\end{proposition}

\begin{proof}
The first two queries are as in the proof of the last proposition.
So we can assume that if we do not already know whether or not $w$
is trivial, $w\in \mathcal{B}_{2}^{11}\cap B_{2}^{01}$ and we may
proceed to the third query. For this, we construct a
\textquotedblleft syllable function\textquotedblright $$f:\{0,1\}\\
\rightarrow \{aa,ab,aB,aA,ba,bb,bB,bA,Ba,Bb,BB,Ba,Aa,Ab,AB,AA\}.$$
It maps $AA$,$BB$,$Aa$,$aA$,$Ab$,$AB$,$ab$,$aB$ to $0$ and $Bb$,$bB$,$%
BA$,$bA$,$Ba$,$ba$,$aa$,$bb$ to $1$. Note that, since $w\in \mathcal{B}%
_{2}^{11}\cap B_{2}^{01}$, $w$ is either a cyclic rotation of
$(AAbb)^{\pm 1} $ or $w$ is an anagram of $AaBb$. Words in the first
case are all trivial, because $a^{2}=b^{2}$ is a relation in $G$,
and these words are all balanced under the syllable function. Words
in the second case are nontrivial if and only if they are nontrivial
commutators. Commutators are constant under the syllable function.
Words in the second case which are trivial (\emph{%
i.e.}, not commutators) are all balanced under the syllable
function. Thus a third query of \textquotedblleft is $w$
syllable-balanced or syllable-constant\textquotedblright\ will
complete the solution of the word problem. The following table lists
all 0-syllabs and 1-syllabs:
\begin{equation*}
\begin{array}{l|l}
\text{0-syllabs} & \text{1-syllabs} \\ \hline
AA & aa \\
BB & bb \\
Aa & Bb \\
aA & bB \\
Ab & bA \\
AB & BA \\
ab & ba \\
aB & Ba%
\end{array}%
\end{equation*}
\end{proof}

\bigskip

While the group $G$ in the last proposition is nonabelian, it can be shown
to have a free abelian subgroup of rank $2$ and index $4$; it is an
extension of $\mathbb{Z}\oplus \mathbb{Z}$ by the Klein 4-group.

\begin{proposition}
Consider the group presented by $G=\langle a,b,c,d\mid a^{2}b^{2}=b^{2}a^{2}\rangle $.
Suppose we are given a word $w$ of length $8$ in $\mathcal{A}$ such that $%
w\in \mathcal{F}_{3}^{adfg}\cap \mathcal{F}_{3}^{abcd}$. Then the quantum query
complexity of the question \textquotedblleft Does $w$ represent the trivial
element of $G$?\textquotedblright\ is at most $3$.
\end{proposition}

\begin{proof}
The first two queries are as in the proof of the last proposition.
So we can assume that if we do not already know whether or not $w$
is trivial, $w\in \mathcal{B}_{3}^{adfg}\cap B_{3}^{abcd}$ and we
may proceed to the third query. For this, we construct an
extended syllable function whose output has a cardinality of
$2^{n-1}$. Some of the elements are listed below:
\[
\begin{tabular}{ll}
$f:\{0,1\}\longrightarrow $ & $\{aaaa,bbbb,BBBB,AAAA,aaab,aabb,abbb,$ \\
& $aaaB,aaBB,aBBB,aaaA,aaAA,aAAA,aaAA,$ \\
& $aAAA,bbbB,bbBB,bBBB,bbbA,bbAA,bAAA,$ \\
& $bbba,bbaa,baaa,BBBa,BBaa,Baaa,BBBb,$ \\
& $BBbb,Bbbb,BBBA,BBAA,BAAA,AAAa,AAaa,$ \\
& $Aaaa,AAAb,AAbb,Abbb,AAAB,AABB,ABBB,...\}$%
\end{tabular}%
\]
Examples of this map are
\begin{eqnarray*}
AAAA,BBBB,Abbb,AAaa,aBBB,aaBB,aaaB,abAB,ABab,ABab,AABB \dots
\text{to 0}
\end{eqnarray*}
and
\begin{eqnarray*}
aaaa,bbbb,Bbbb,BBbb,bBBB,bAAA,BBAA,BAAA,baaa,bAAA,aabb,bbaa \dots
\text{to 1}.
\end{eqnarray*}
Note that since $w\in \mathcal{B} _{3}^{adfg}\cap B_{3}^{abcd}$,
$w$ is either a cyclic rotation of $(AABBaabb)^{\pm 1} $ or $w$ is
an anagram of $AAaaBBbb$. Words in the first case are all trivial,
because $a^{2}b^{2}=b^{2}a^{2}$ is a relation in $G$, and these
words are all balanced under the syllable function. Words in the
second case are nontrivial if and only if are nontrivial
sequence of letters, that is not commutator-like sequence
with respect to the presentation.
Words in the second case which are trivial (\emph{%
i.e.}, not trivial sequence) are all balanced under the extended
syllable function. Thus a third query of \textquotedblleft is $w$
syllable-balanced or syllable-constant\textquotedblright\ will
complete the solution of the word problem.
\end{proof}

\bigskip

The same considerations can be made by looking at different sets of
generators or relations like $G=\langle a,b,c,d\mid
c^{2}d^{2}=d^{2}c^{2}\rangle $ and $G=\langle a,b,c,d\mid
b^{2}c^{2}=c^{2}b^{2}\rangle $. It is important to notice that all the alternate sets of relations five groups isomorphic to the group considered in Proposition 5. To see this, it is sufficient to relabel the generators. The relation in $G$ is in fact very general and it is possible to obtain the same result with a whole family of similar relations. This can be done by varying the parity function used for the queries, choosing the presentation accordingly. Moreover such a group is a free group of rank $2$ with $G=\langle a,b,c,d\mid
a^{2}b^{2}=b^{2}a^{2}\rangle $. It is simple to see that since the other two generators, $c$ and $d$, are not involved in the proof, it is possible to take the free product of $G$ with any free group and get to the same conclusion. In particular it is possible to extend the free product with \emph{any} group and see the invariance of those three quantum queries under free products.

Notice that the choice of some particular kind
of relations and an higher number of generators in the setting
of the problem may increase the number of queries required. The reason of this is the exponential growth in the number of permutations, in
particular, in those cases where splitting the words in parity
balanced and parity constant does not help. Generalize to other
different sets of generators and possibly for free products, and
limiting to commutator words might give interesting promises.

\section{Conclusions}

We have extended the original Deutsch-Josza algorithm to functions
of arbitrary length binary output, and we have introduced a
more general concept of parity. The setting described allows us to consider maps between binary strings and alphabet of various length. In the quantum regime, some instances of the word problem for small alphabets and free groups, can be solved in a reduced number of queries with respect to the deterministic classical case. Extensions to more general groups and presentations may give interesting promises. It is not clear that the success of procedures similar to the ones discussed here depends or not on the group considered. We have seen that the $X$-parity of a function, for some fixed set of binary strings $X$, can be determined with the Deutsch-Jozsa procedure, when $X$ consists of an appropriate subgroup (in our examples, a subgroup of index two). It has to be verified that the toy problems considered here can be re-interpreted as instances of the Abelian Hidden Subgroup Problem. In such a case, the problems could be solved with a slightly different technique, but with essentially the same number of oracle queries.

\bigskip

\emph{Acknowledgments} The authors would like to thank Andrew Childs for useful remarks. Part of this work as been done while Andrea Casaccino was attending \textquotedblleft The Seventh Canadian Summer School on Quantum Information\textquotedblright, hosted by the Perimeter Institute for Theoretical Physics and the Institute for Quantum Computing at the University of Waterloo.

\end{document}